\documentclass{article}

\pdfoutput=1

\usepackage{amsmath,amsthm,amsfonts,amssymb}

\usepackage{graphicx,pstricks}

\newtheorem{theorem}{Theorem}
\newtheorem{corollary}[theorem]{Corollary}
\theoremstyle{definition}
\newtheorem{definition}[theorem]{Definition}
\theoremstyle{plain}
\newtheorem*{pruningscan}{Pruning scan}

\newcommand{\hshearbox}[3]{\scalebox{0.866025}[#2]{\rotatebox{210}%
{\scalebox{1.73205}[-0.57735]{\rotatebox{60}{\scalebox{-1.1547}[#1]{#3}}}}}}
\newcommand{\vshearbox}[3]{\scalebox{#2}[0.866025]{\rotatebox{210}%
{\scalebox{-0.57735}[1.73205]{\rotatebox{60}{\scalebox{#1}[-1.1547]{#3}}}}}}

\newcommand{\tilenamebox}[1]{\rput[B]{0}(0,-2.5pt){$#1$}}
\newcommand{\lockstack}[2]{\begin{array}{@{}c@{}}\begin{picture}(42,40)(4,0)
\put(25,35){\hshearbox{0.333333}{1}{\tilenamebox{#1}}}
\put(20,22.5){\tilenamebox{#1}}
\put(40,27.5){\vshearbox{-0.333333}{-1}{\tilenamebox{#1}}}
\put(20,07.5){\tilenamebox{#2}}
\put(40,12.5){\vshearbox{-0.333333}{-1}{\tilenamebox{#2}}}
\multiput(5,0)(0,15){3}{\line(1,0){30}\line(1,1){10}}
\put(05,30){\line(0,-3){30}\line(1,1){10}}
\put(35,30){\line(0,-3){30}}
\put(15,40){\line(1,0){30}\line(0,-1){30}}
\multiput(9,31)(2,0){13}{\hshearbox{1}{1}{\makebox(0,0)[lb]{\line(0,1){1.5}}}}
\multiput(17,39)(2,0){13}{\hshearbox{1}{1}{\makebox(0,0)[rt]{\line(0,-1){1.5}}}}
\end{picture}\end{array}}
\newcommand{\tilestack}[3]{\begin{array}{@{}c@{}}\begin{picture}(42,55)(4,0)
\put(25,50){\hshearbox{0.333333}{1}{\tilenamebox{#1}}}
\put(20,37.5){\tilenamebox{#1}}
\put(40,42.5){\vshearbox{-0.333333}{-1}{\tilenamebox{#1}}}
\put(20,22.5){\tilenamebox{#2}}
\put(40,27.5){\vshearbox{-0.333333}{-1}{\tilenamebox{#2}}}
\put(20,07.5){\tilenamebox{#3}}
\put(40,12.5){\vshearbox{-0.333333}{-1}{\tilenamebox{#3}}}
\multiput(5,0)(0,15){4}{\line(1,0){30}\line(1,1){10}}
\put(05,45){\line(0,-3){45}\line(1,1){10}}
\put(35,45){\line(0,-3){45}}
\put(15,55){\line(1,0){30}\line(0,-1){45}}
\multiput(9,46)(2,0){13}{\hshearbox{1}{1}{\makebox(0,0)[lb]{\line(0,1){1.5}}}}
\multiput(17,54)(2,0){13}{\hshearbox{1}{1}{\makebox(0,0)[rt]{\line(0,-1){1.5}}}}
\end{picture}\end{array}}
\newcommand{\rmtilestack}[3]{\begin{array}{@{}c@{}}\begin{picture}(50,55)(0,0)
\put(25,50){\hshearbox{0.333333}{1}{\tilenamebox{#1}}}
\put(20,37.5){\tilenamebox{#1}}
\put(40,42.5){\vshearbox{-0.333333}{-1}{\tilenamebox{#1}}}
\put(20,22.5){\tilenamebox{#2}}
\put(40,27.5){\vshearbox{-0.333333}{-1}{\tilenamebox{#2}}}
\put(20,07.5){\tilenamebox{#3}}
\put(40,12.5){\vshearbox{-0.333333}{-1}{\tilenamebox{#3}}}
\multiput(5,0)(0,15){4}{\line(1,0){30}\line(1,1){10}}
\put(05,45){\line(0,-3){45}\line(1,1){10}}
\put(35,45){\line(0,-3){45}}
\put(15,55){\line(1,0){30}\line(0,-1){45}}
\multiput(9,46)(2,0){13}{\hshearbox{1}{1}{\makebox(0,0)[lb]{\line(0,1){1.5}}}}
\multiput(17,54)(2,0){13}{\hshearbox{1}{1}{\makebox(0,0)[rt]{\line(0,-1){1.5}}}}
\put(00,2.5){\line(1,1){50}}
\put(50,2.5){\line(-1,1){50}}
\end{picture}\end{array}}
\newcommand{\tilerow}[4]{\begin{array}{@{}c@{}}\begin{picture}(222,55)(4,0)
\put(25,20){\hshearbox{0.333333}{1}{\tilenamebox{#1}}}
\put(55,20){\hshearbox{0.333333}{1}{\tilenamebox{#2}}}
\put(85,20){\hshearbox{0.333333}{1}{\tilenamebox{#3}}}
\multiput(115,20)(30,0){4}{\hshearbox{0.333333}{1}{\tilenamebox{#4}}}
\put(20,07.5){\tilenamebox{#1}}
\put(50,07.5){\tilenamebox{#2}}
\put(80,07.5){\tilenamebox{#3}}
\multiput(110,07.5)(30,0){4}{\tilenamebox{#4}}
\put(220,12.5){\vshearbox{-0.333333}{-1}{\tilenamebox{#4}}}
\multiput(5,15)(30,0){8}{\line(0,-1){15}\line(1,1){10}}
\put(5,0){\line(1,0){210}\line(1,1){10}}
\put(5,15){\line(1,0){210}}
\put(225,25){\line(0,-1){15}\line(-1,0){210}}
\multiput(0,0)(30,0){7}{%
\multiput(9,16)(2,0){13}{\hshearbox{1}{1}{\makebox(0,0)[lb]{\line(0,1){1.5}}}}}
\multiput(0,0)(30,0){7}{%
\multiput(17,24)(2,0){13}{\hshearbox{1}{1}{\makebox(0,0)[rt]{\line(0,-1){1.5}}}}}
\end{picture}\end{array}}

\title{Solving Mahjong Solitaire boards with peeking}
\author{Michiel de Bondt}

\allowdisplaybreaks

\begin{document}

\maketitle

\begin{abstract}
\noindent
We first prove that solving Mahjong Solitaire boards with peeking is NP-complete,
even if one only allows isolated stacks of the forms $aab$ and $abb$. We 
subsequently show that layouts of isolated stacks of heights one and two can always
be solved with peeking, and that doing so is in P, as well as finding an optimal
algorithm for such layouts without peeking.

Next, we describe a practical algorithm for solving Mahjong Solitaire 
boards with peeking, which is simple and fast. The algorithm uses an effective
pruning criterion and a heuristic to find and prioritize critical groups.
The ideas of the algorithm can also be applied to solving Shisen-Sho with peeking.
\end{abstract}

\noindent
Mahjong Solitaire is a game which is played with the 144 tiles of the Chinese
game Mahjong. The tiles are distributed in 36 groups of four tiles each. In the
beginning of the game, the tiles are stacked randomly in a predefined pattern, called
the \emph{layout}. The so-called {\em turtle layout} is used the most and therefore
called the default layout as well. After stacking the tiles, the object is to remove 
all tiles under certain rules. These rules are as follows.
\begin{itemize}

\item A tile is {\em playable}, if and only if no other tile is lying upon it, not 
even partially, and either its left side or its right side does not touch any other tile.

\item Only playable tiles may be played, but solely in pairs of tiles of the same
group. Thus removing all tiles takes 72 removals of pairs of similar tiles.

\end{itemize}
During game play, one cannot see tiles which are completely below other tiles. 
Sometimes a tile can be seen partially, namely where it is not covered by 
an other tile.

\section{Motivation}

Since the theoretical content of this article is not enough to justify 
publication, the idea of writing this article came somewhat later and is solely
motivated by the experience that people try to write similar algorithms for
both Mahjong Solitaire and Shisen-Sho, but fail to get them fast enough. 
Hence the theory in this article fulfills the main purpose of science, 
namely serving practice. 

The sensitivity of the performance of a program for Mahjong Solitaire to 
design is connected to the NP-completeness of the problem. For that reason, we
include an NP-completeness result as well. Although such a result has not been 
published yet in an official forum like this, the result that Mahjong Solitaire
with peeking is NP-complete is not new. For that reason, our result will be the
novelty that Mahjong Solitaire with peeking is already NP-complete when the layout 
only contains isolated stacks of height three.

\section{Complexity results}

When we generalize the number of groups from 36 to any natural number, we
get a game which can be subjected to complexity analysis. The following results
are known.

\begin{theorem}[A. Condon, J. Feigenbaum, C. Lund, and P. Shor, 
{\cite[Theorem 3.6]{nopeek}}]
Mahjong Solitaire is PSPACE-complete.
\end{theorem}

\begin{theorem}[D. Eppstein, {\cite[Shanghai]{peek}}] \label{peekth}
Mahjong Solitaire is NP-complete when peeking is allowed.
\end{theorem}

\noindent
We only consider the variant with peeking of Mahjong Solitaire in this article.
In order to be able to include a complexity result and to be original at the
same time, we refine theorem \ref{peekth}, which results in the below theorem.

\begin{theorem} \label{mahsolNP}
Solving Mahjong Solitaire boards with peeking is NP-complete, even 
if one only allows isolated stacks of the forms
$$
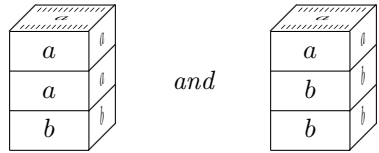
\tilestack{a}{a}{b} \qquad \mbox{and} \qquad \tilestack{a}{b}{b}
$$
(We write the tile names on the sides to make peeking easier.)
\end{theorem}

\begin{proof}
We reduce from 3-SAT. Our construction consists of three steps.
\begin{itemize}

\item \emph{Step 1: initial setup.} \\
We set up the following eight tile stacks.
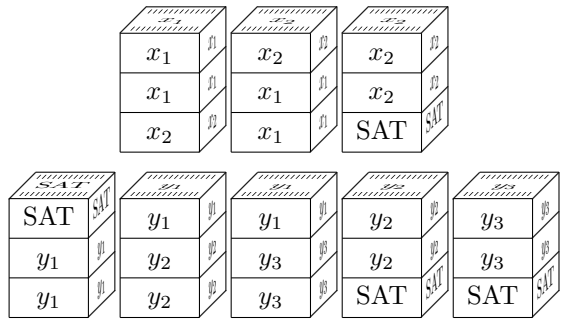
\begin{gather*}
\tilestack{x_1}{x_1}{x_2}\tilestack{x_2}{x_1}{x_1}\tilestack{x_2}{x_2}{\textrm{SAT}} \\
\tilestack{\textrm{SAT}}{y_1}{y_1}\tilestack{y_1}{y_2}{y_2}\tilestack{y_1}{y_3}{y_3}
\tilestack{y_2}{y_2}{\textrm{SAT}}\tilestack{y_3}{y_3}{\textrm{SAT}}
\end{gather*}
For each $x$-group $x_i$, there is exactly one stack with two $x_i$-tiles on top. 
In order to release the third tile of such a stack, it is necessary to free
the other two $x_i$-tiles. This property will be maintained in the other steps.

For each $y$-group $y_j$, there is exactly one stack with two $y_j$-tiles below. 
In order to release tiles below one of the other two $y_j$-tiles, it is 
necessary to remove the top tile of the stack with two $y_j$-tiles below first, 
since otherwise both $y_j$-tiles of this stack cannot be freed any more. This 
property will be maintained in the other steps as well.

The SAT group is a key group. The ultimate problem (which corresponds to solving
the instance of 3-SAT) will be to remove the first pair
of the SAT group, which can only consist of the SAT tile below two tiles of an 
$x$-group and the SAT tile above the two $y_1$-tiles, since the other two SAT-tiles
are blocked indirectly by the SAT tile which is above the two $y_1$-tiles. 
After removing the first SAT pair, it will be easy to clear the board. 

\item \emph{Step 2: adding variables.} \\
For each variable $V_i$, we perform the following. Let $m$ be the
largest index of the $x$-groups and $n$ be the largest index of the
$y$-groups. Remove two tile stacks and add fourteen
tile stacks as drawn below.
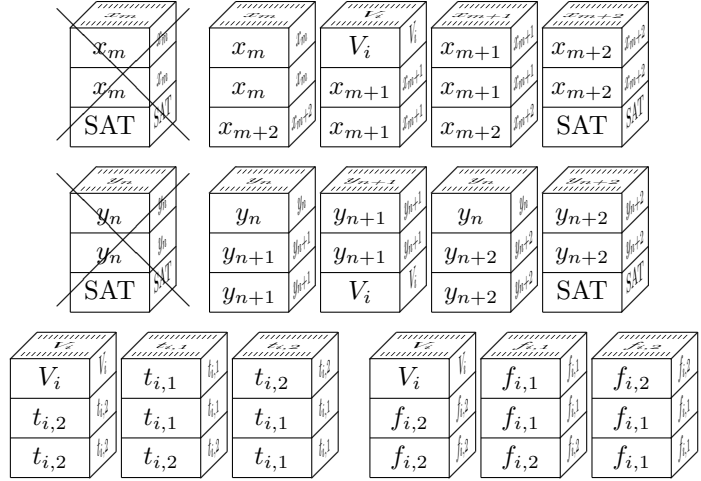
\begin{gather*}
\rmtilestack{x_m}{x_m}{\textrm{SAT}} ~~ \tilestack{x_m}{x_m}{x_{m+2}}
\tilestack{V_i}{x_{m+1}}{x_{m+1}} \tilestack{x_{m+1}}{x_{m+1}}{x_{m+2}}
\tilestack{x_{m+2}}{x_{m+2}}{\textrm{SAT}} \\
\rmtilestack{y_n}{y_n}{\textrm{SAT}} ~~ \tilestack{y_n}{y_{n+1}}{y_{n+1}}
\tilestack{y_{n+1}}{y_{n+1}}{V_i} \tilestack{y_n}{y_{n+2}}{y_{n+2}}
\tilestack{y_{n+2}}{y_{n+2}}{\textrm{SAT}} \\
\tilestack{V_i}{t_{i,2}}{t_{i,2}} \tilestack{t_{i,1}}{t_{i,1}}{t_{i,2}} 
\tilestack{t_{i,2}}{t_{i,1}}{t_{i,1}} \quad
\tilestack{V_i}{f_{i,2}}{f_{i,2}} \tilestack{f_{i,1}}{f_{i,1}}{f_{i,2}}
\tilestack{f_{i,2}}{f_{i,1}}{f_{i,1}}
\end{gather*}
One can show by induction that in order to release the first SAT pair,
one must remove the $V_i$-tile which is above two tiles of an $x$-group,
for all $i$. Furthermore, one can only remove the tile $V_i$ below two tiles 
of a $y$-group after the first SAT pair is freed, for all $i$. 

Thus for each $i$,
one has to choose between the $V_i$ above two tiles of a $t_i$-group and the $V_i$ 
above two tiles of a $f_i$-group, for removal along with the $V_i$-tile above two 
tiles of an $x$-group.
Removing the $V_i$-tile above two tiles of a $t_i$-group first corresponds to setting
$V_i$ to \emph{true} and removing the $V_i$-tile above two tiles of a $f_i$-group first corresponds to setting
$V_i$ to \emph{false}.

For each $t_i$-group and each $f_i$-group, there is exactly one stack with
two tiles of the group below, just like with each $y$-group.  This property is
maintained in step 3. Thus the $t_i$-groups 
and $f_i$-groups act in a similar manner as the $y$-groups.

\item \emph{Step 3: adding clauses.} \\
For each clause $C_j$, we perform the following. Let $m$ be the
largest index of the $x$-groups. First, we remove one tile stack and 
add four tile stacks as drawn below.
$$
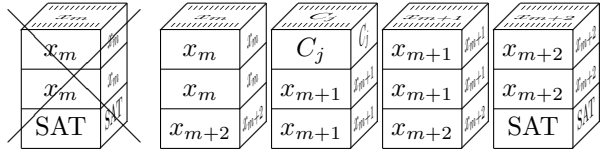
\rmtilestack{x_m}{x_m}{\textrm{SAT}} ~~ \tilestack{x_m}{x_m}{x_{m+2}}
\tilestack{C_j}{x_{m+1}}{x_{m+1}} \tilestack{x_{m+1}}{x_{m+1}}{x_{m+2}}
\tilestack{x_{m+2}}{x_{m+2}}{\textrm{SAT}}
$$
Subsequently, for each of the at most three variables $V_i$ in $C_j$, we do the following.
\begin{itemize}

\item If $V_i$ appears in a \emph{positive} literal, let $k$ be
the largest index of the $t_i$-groups. Remove one tile stack and add
four tile stacks as drawn below.
$$
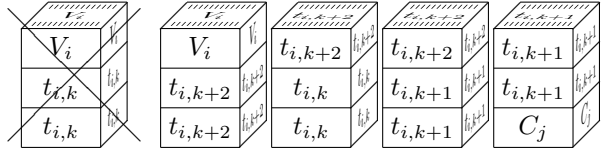
\rmtilestack{V_i}{t_{i,k}}{t_{i,k}} ~~ \tilestack{V_i}{t_{i,k+2}}{t_{i,k+2}}
\tilestack{t_{i,k+2}}{t_{i,k}}{t_{i,k}} \tilestack{t_{i,k+2}}{t_{i,k+1}}{t_{i,k+1}}
\tilestack{t_{i,k+1}}{t_{i,k+1}}{C_j}
$$

\item If $V_i$ appears in a \emph{negative} literal, let $k$ be
the largest index of the $f_i$-groups. Remove one tile stack and add
four tile stacks as drawn below.
$$
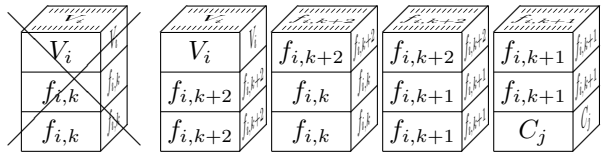
\rmtilestack{V_i}{f_{i,k}}{f_{i,k}} ~~ \tilestack{V_i}{f_{i,k+2}}{f_{i,k+2}}
\tilestack{f_{i,k+2}}{f_{i,k}}{f_{i,k}} \tilestack{f_{i,k+2}}{f_{i,k+1}}{f_{i,k+1}}
\tilestack{f_{i,k+1}}{f_{i,k+1}}{C_j}
$$

\end{itemize}
If we repeat adding literals to make the number of times a variable is treated 
equal to three for each clause 
(in case the definition of 3-SAT does not require exactly three literals), 
then for each $j$, all four $C_j$-tiles are used.
 
One can show by induction that in order to release the first SAT pair,
one must remove the $C_j$-tile which is above two tiles of an $x$-group,
for all $j$. To do so, one of the other three $C_j$-tiles has to be released.

For each literal of $C_j$ which evaluates to true (with respect to the choice of 
the pair to be removed first in the group of the literal's variable), 
one $C_j$-tile can be released.
Other $C_j$-tiles cannot be freed before removing the first SAT-pair.
Thus the $C_j$-tile which is above two tiles of an $x$-group can only be released
if $C_j$ evaluates to true.

\end{itemize}
In order to remove the first two SAT-tiles, all variables $V_i$ must be assigned
boolean values and all clauses $C_j$ must be satisfied.
After removing the first SAT-pair, all remaining $V_i$-tiles can be released, after which
all remaining $C_j$-tiles can be freed. All other tiles are removed along with this.
Hence our reduction from 3-SAT is complete.
\end{proof}

\noindent
With one level Mahjongg solitaire, all tiles are on the same level and can therefore
already be seen without peeking.

\begin{corollary}
One level Mahjong Solitaire is NP-complete.
\end{corollary}

\begin{proof}
Replace each stack of theorem \ref{mahsolNP} by a row of tiles, as follows.
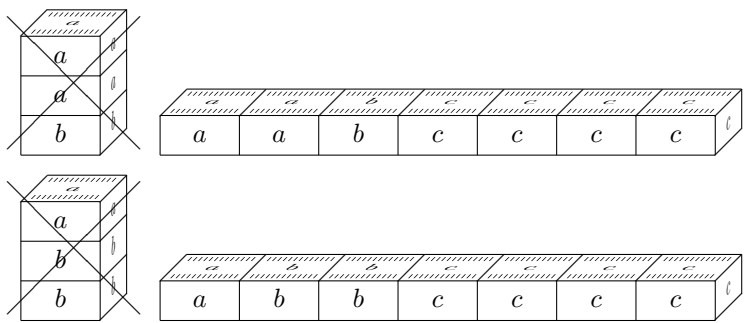
\begin{gather*}
\rmtilestack{a}{a}{b} ~~ \tilerow{a}{a}{b}{c} \\
\rmtilestack{a}{b}{b} ~~ \tilerow{a}{b}{b}{c}
\end{gather*}
Here, every stack of theorem \ref{mahsolNP} gets a new tile group $c$, of course. 
\end{proof}

\noindent
Now that we have NP-completeness with stacks of height three, it seems
natural to look what happens with stack of smaller heights.

\begin{theorem}
If the layout only consists of isolated stacks of heights one and two, 
then one can always remove all tiles of the Mahjong Solitaire game when peeking 
is allowed. 

Without peeking, an optimal strategy for a Mahjong Solitaire layout as above 
is to repeat the following. First choose an arbitrary group which has a match, 
i.e.\@ a pair of playable tiles. Next remove a match of that group with the maximum
number of tiles not on the ground.
\end{theorem}

\begin{proof}
Assume the layout only consists of isolated stacks of heights one and two. 
Since there are four 
tiles of each group, there cannot be a {\em blocked cycle}, see definition 
\ref{blockedcycle} below. Hence by theorem \ref{mahsol2} below, all tiles can be 
removed when peeking is allowed.

Without peeking, one can only remove the wrong tiles of a group when exactly three
of the four tiles are not on the ground, in case the formulated strategy is obeyed. 
But in that case, there is no information
that might lead to a best guess, since all three stacks of height two with
a tile of the above group on top look the same.
\end{proof}

\noindent
When playing Mahjong Solitaire, the number of tiles of each group
does not stay equal to four, but will become two first and zero later when
all tiles are removed. Hence we formulate the following definition.

\begin{definition} \label{blockedcycle}
Assume the layout only consists of isolated stacks of heights one and two,
and the number of tiles of each group is two or four. A \emph{blocked
cycle} is a subset $\{p_1, p_2, \ldots, p_k\}$ of distinct groups of size {\em two}
which are stacked as follows, where $k \ge 1$.
$$
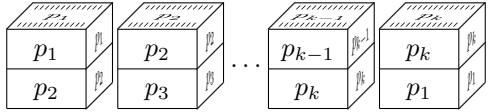
\lockstack{p_1}{p_2} \lockstack{p_2}{p_3} \ldots 
\lockstack{p_{k-1}}{p_k} \lockstack{p_k}{p_1}
$$
In the somewhat degenerate case $k = 1$, there is only one stack, 
one with two $p_1$-tiles.
\end{definition}

\noindent
Notice that the board cannot be solved when a blocked cycle is present.
By the below theorem, we see how to play with stacks
of heights one and two when peeking is allowed.

\begin{theorem} \label{mahsol2}
If the layout only consists of isolated stacks of heights one and two, and all groups
have two or four tiles, then one can remove all tiles of the Mahjong Solitaire 
game when peeking is allowed, if and only if no blocked cycle is present.

Furthermore, for each group which has a match, 
at least one of the matches can be played without introducing a new blocked cycle.
\end{theorem}

\begin{proof}
We first prove the second claim.
Assume we have a layout of isolated stacks of heights one and two, 
and that some group $g$ has a match
such that playing this match will result in a new blocked cycle. 

Then there will be exactly one
tile of $g$ on the ground after playing that match. This can be
counteracted by playing another match of $g$, except if exactly three
tiles of $g$ are not on the ground. But in that case, every group which is
in the blocked cycle, except $g$, is already reduced to only two tiles,
of which one is on the ground, before removing the match of $g$. 

Hence there is only one possible blocked cycle with tiles of $g$ up to 
cyclic shift.
Thus only one of the three tiles of $g$ which are not on the ground is within
a blocked cycle after playing the other two.
Hence the blocked cycle can be broken by playing the only playable tile
of $g$ within the blocked cycle together with another tile of $g$ instead. 
This proves the second claim of theorem \ref{mahsol2}.

To prove the first claim, notice first that by the second claim,
one can avoid introducing the presence of a blocked cycle as long as 
there is a match. Thus assume that there are no matches available any 
more.
 
If more than half of the tiles are on top of a stack,
then there will be a match, because all groups have size two at least.
Furthermore, there will be a match when a group of size four is present,
because at least half of the tiles is on top of a stack.

Thus there are only stacks of height two and only groups of size two, and
one can easily see that all stacks are contained in a blocked cycle. 
This completes the proof of theorem \ref{mahsol2}.
\end{proof}

\section{A practical algorithm}

To determine whether a given board is solvable, a first idea may be to try
and play any match and next determine the solvability recursively. If the board
is not solvable any more after playing the first match, then try another match 
instead and test it by going into recursion again, until all possible matches 
are tried.

But this approach is way too slow. One way to improve the performance is to 
apply {\em cleaning operations} initially (and thus implicitly in the recursion 
as well). A first cleaning operation is to play all remaining (two or four)
tiles of a group when this is possible. Another cleaning operation is the following.
If three tiles of a group are playable and at least one of them does not block 
any other tile, then play the other two.

Another speedup is obtained by taking into account that two matches of different groups
can be played in any order, and that the order does not affect the solvability. For this
purpose, one can order the groups and require that matches which are playable at some
stage are played in increasing order of the groups. Thus if one is about to play a match 
within a group $g$, all current matches of groups lower than $g$ will be forbidden in the 
branch of the search tree after playing the above match within $g$.

Notice that every group has six possible matches. But when playing a match of two tiles of a group is forbidden, it does not make sense when the other two tiles of that group are still
played together. Thus there are only three possible matches to distinguish. In other words,
there are three possible {\em pairings} of a group of four tiles, namely 
$\{\{1,2\},\{3,4\}\}$, $\{\{1,3\},\{2,4\}\}$ and $\{\{1,4\},\{2,3\}\}$. When a match within
a group $g$ is played, a pairing of $g$ is chosen and for groups of lower order than $g$,
one pairing is marked as unallowed in case such a group has exactly one match. If a group 
of lower order than $g$ has more than one match, then all pairings of that group will get
forbidden, which means that group $g$ is not the right group to play at the moment.

This is more or less the algorithm of \cite{gimeno}. Although the algorithm is not very 
fast, it can be used to evaluate initial turtle layouts, most of them within a coffee break, 
but some of them take more than a day. Others did similar things, and subsequently
applied low-level improvements, include writing efficient assembly code and searching with multiple theads.

This is however not the best idea at this stage, because high-level 
improvements are far more effective and still available. We discuss two of them
which lead to an algorithm that solves a typical board in less than a second 
and a worst-case board in about a minute.

\subsection{A pruning criterion}

The above cleaning operations can be seen as a pruning criterion, since they essentially spot
groups whose removal is not problematic in any sense and branching within such groups is eliminated. A more direct pruning criterion which can be tried before going into recursion 
is the following.

\begin{pruningscan}
Try if you can clean the board by playing the first two tiles of each group simultaneously (choose any match), and playing the third and fourth tile individually (in case the group has four tiles). 
\end{pruningscan}

\noindent
The rules in the pruning scan are not only more gracious than the original rules, 
but it is also impossible to play the wrong match. Therefore, 
if the pruning scan has a negative answer, then the board will not be solvable 
and one can search further without going into recursion.

Furthermore, for groups where cleaning can be applied upon, this pruning criterion 
works as well, because such groups do not affect the result of the above pruning 
scan. But subsequent optimizations, which will be discussed in the next subsection,
will counteract the cleaning operations.

\subsection{Critical groups}

Instead of choosing matches to play recursively, one can also choose group pairings
recursively. With 36 groups, one gets a search tree of size $3^{36}$, since there
are three possible pairings for each group. The advantage of a group-directed
search tree with respect to a match-directed search tree is that one can choose the
order of groups in a group-directed search tree in an arbitrary manner. With a
match-directed search tree, the order of groups corresponds to the order of 
playing of the first pair of tiles of them.

With a group-directed search tree, no tiles are actually removed, which makes
that the prune scans will not get smaller. But a scan of a board with half of the tiles
removed in only a factor two faster than a scan with all tiles, which is peanuts
in this context. Furthermore, groups which are associated 
a pairing are removed corresponding to that pairing in a pruning scan and no individual 
tile removal are allowed for them, otherwise one does not progress in the search.

To improve the effect of the pruning, we try to choose a critical group and 
prioritize it by expanding it in the search tree (on top of the current stage). 
One way to find a critical group is to choose a group that
allows the minimum number of pairings in the pruning scan. If all unpaired
groups allow all three pairings, then one can try to find a group whose four tiles
cannot be removed simultaneously during a pruning scan.

This heuristic for finding critical groups is a significant improvement of the
algorithm, but some boards still take quite some time. For that reason, we present
a totally different heuristic for finding critical groups, one which appears even 
better in practice since all boards seem to be solved within a reasonable amount of time.

Suppose that we are either on top of the search, or that the last group which is 
assigned a pairing has a pairing greater than one. We choose the next group
without any heuristic yet, and assign it pairing one. Next, we perform a pruning 
scan. If the pruning scan is effective, then we decide that the chosen group is 
critical, assign it pairing two (the next pairing), and go further into recursion.

If the pruning scan is not effective, then we add another group with pairing one to
the search tree, and repeat adding groups until we get an effective pruning scan (or
a solution). When we hit an effective pruning scan, we assign the last added group 
pairing two just as above, but we do the following with all groups that we just added to 
the search tree with pairing one. We remove all such groups that are not needed for the 
pruning scan to be effective from the search tree. 

Since there might be more minimal combinations of groups that keep the pruning scan 
effective, we have to be more specific here.
We remove the groups in backward order, starting from the last added group
(which is of course necessary for the pruning), going up to the previous group 
with pairing larger than one (or the beginning of the search tree), from where we started adding
groups with pairing one. Thus we adapt the search tree 
backwards. For the groups that are removed from the search tree, we cycle tiles
two, three and four either forward or backward, hoping that they appear critical
later on due to a different first pairing. 

A more progressive rearrangement strategy on the search tree 
(based on maintaining a pruning), where the progress in the 
ternary search tree (which is a ternary fraction) is maximized, appeared 
overdone with respect to speed. 

\subsection{Random solving}

One can also try to solve the board by chosing matches randomly. This is likely
to fail, but then one can repeat the attempt another time. In the current 
implementation of the solver, $\lfloor 1.2^{36} \rfloor = 708$ attempts are done 
to solve a board randomly (since there are $36$ groups which have still four tiles
initially). The random attempts are combined with the above cleaning heuristic,
but first, a pruning scan on top level is performed, since some layouts
are so impossible that almost all boards can be discarded with such a scan.

Already twenty years ago, Ken McDonald wrote a solver which tries to solve boards
by random match selection, see \cite{mjsolv12}. He also remarks that the difficulty 
of solvable boards can be measured by the fraction of succesful attempts of solving.

\subsection{Layouts and probabilities}

Below follow probabilities that a random board of some layout of Xmahjongg 3 
cannot be won even when peeking is allowed. All results are based on scans of 
100,000 boards, except the default layout which is based on a scan of 10,000,000 
boards. This latter scan took about 40 hours on a single thread of a Xeon L5420.
From the other layouts, the ox took with about 110 minutes the most time, quickly
followed by the bridge layout, which has significantly more really hard boards 
than any of the other layouts. 

\begin{center}
\begin{tabular}{lll}
Default: 2.95\% \\
\phantom{Hourglass: 100\% (all)} &
\phantom{Hourglass: 100\% (all)} &
\phantom{Hourglass: 100\% (all)} \\
Arena: 2.6\% &      Farandole: 7.9\% &       Rat: 5.1\% \\ 
Arrow: 8.2\% &      Hare: 18\% &             Rooster: 22\% \\
Boar: 4.7\% &       Horse: 20\% &            Snake: 4.2\% \\
Bridge: 32\% &      Hourglass: 100\% (all) & Theater: 0.62\% \\
Ceremonial: 1.8\% & Monkey: 9.9\% &          Tiger: 22\% \\ 
Deepwell: 6.0\% &   Ox: 47\% &               Wedges: 4.8\% \\
Dog: 7.3\% &        Papillon: 100\% (all) \\
Dragon: 7.5\% &     Ram: 6.9\% \\
\end{tabular}
\end{center}

\noindent
The results for transposed layouts and the source can be found on the authors homepage
\cite{mjsolver}, as well as a DLL for solving and a hacked version of Xmahjongg 3.6.1 
with the solver.


\begin{thebibliography}{9}

\bibitem{mjsolver}
Michiel de Bondt, Solitaire Mahjongg solver, \\
\verb+http://www.math.ru.nl/+\makebox(5,0){\~{ }\,}\verb+debondt/mjsolver.html+

\bibitem{nopeek}
A. Condon, J. Feigenbaum, C. Lund, and P. Shor, Random debaters and the hardness of
approximating stochastic functions, SIAM Journal on Computing 26:2 (1997) 369--400.

\bibitem{peek}
David Eppstein, Computational Complexity of Games and Puzzles, \\
\verb+http://www.ics.uci.edu/+\makebox(5,0){\~{ }\,}\verb+eppstein/cgt/hard.html#shang+

\bibitem{gimeno}
Pedro Gimeno Fortea, Mahjongg Solitaire Solver, \\
\verb+http://www.formauri.es/personal/pgimeno/mj/mjsol.html+

\bibitem{mjsolv12}
Ken McDonald, MJSolver v1.2, \\
\verb+http://cd.textfiles.com/ugameware/MAJONG1/MJSOLV12.ZIP+

\end{thebibliography}
\end{document}